\documentclass{article}
\usepackage{kr}

\usepackage{float}

\def\phi{\varphi}

\usepackage{times}
\usepackage{soul}
\usepackage{url}
\usepackage[hidelinks]{hyperref}
\usepackage[utf8]{inputenc}
\usepackage[small]{caption}
\usepackage{graphicx}
\usepackage{amsmath}
\usepackage{amsthm}
\usepackage{amssymb}
\usepackage{booktabs}
\usepackage{algorithmic}
\usepackage{tikz}
\usetikzlibrary{arrows, automata, quotes, arrows.meta, fit, patterns}
\usetikzlibrary{decorations.pathreplacing,trees,arrows}
\usetikzlibrary{positioning,arrows,calc}
\usepackage{synttree, bussproofs,scrextend, stmaryrd, rotating, xcolor}
\usepackage{esvect,pgf, tikz, color}
\usetikzlibrary{arrows, automata, positioning}
\usepackage{comment}
\usepackage[all]{xy}
\usepackage[inline]{enumitem}
\usepackage{upgreek}
\usepackage[capitalise]{cleveref}

\usepackage[linesnumbered,ruled,lined,boxed]{algorithm2e}

\providecommand{\acknowledgments}[1]{\textbf{Acknowledgments.} #1}

\newcommand{\myDots}{\ifmmode\mathinner{\ldotp\kern-0.1em\ldotp\kern-0.1em\ldotp}\else.\kern-0.1em.\kern-0.3em.\fi}

\newcommand{\lem}{Lem.}
\newcommand{\fig}{Fig.}
\newcommand{\thm}{Thm.}
\newcommand{\sect}{Sect.}
\newcommand{\dfn}{Def.}
\newcommand{\alg}{Alg.}
\newcommand{\app}{App.}
\newcommand{\ifandonlyif}{iff }
\newcommand{\iffi}{\textit{iff} }

\definecolor{tim}{RGB}{0, 0, 250}

\newcommand{\conp}{\mathrm{CoNP}}
\newcommand{\np}{\mathrm{NP}}
\newcommand{\pspace}{\mathrm{PSPACE}}

\newcommand{\cf}{f_{\ast}}
\newcommand{\clr}[1]{\widetilde{#1}}
\newcommand{\size}[1]{|#1|}
\newcommand{\card}[1]{|#1|}

\newcommand{\len}[1]{|#1|} 
\newcommand{\sufo}[1]{\mathrm{sufo}(#1)}
\newcommand{\prc}[1]{\mathrm{pcs}(#1)}

\newcommand{\bigo}{\mathcal{O}}

\newcommand{\sar}{\vdash}

\newcommand{\ns}{\Phi}
\newcommand{\nsii}{\Psi}
\newcommand{\sprset}{\Gamma}
\newcommand{\nseq}{\Delta}
\newcommand{\sep}{\vdash} 

\newcommand{\sprel}{\preceq^{*}_{\sprset}}

\newcommand{\langv}{\mathcal{L}_{\vocab}}
\newcommand{\langsharp}{\mathcal{L}_{\preceq}}
\newcommand{\spbox}[1]{\Box\nolimits_{#1}}
\newcommand{\spdia}[1]{\Diamond\nolimits_{#1}}
\newcommand{\sharper}{\preceq}

\newcommand{\spv}{\PropStandpointLogic(\vocab)}

\def\id{(id)}
\def\conr{(\land)}
\def\disr{(\lor)}
\newcommand{\spboxr}[1]{(\Box_{#1})}
\newcommand{\spdiar}[1]{(\Diamond_{#1})}
\newcommand{\spdiari}[1]{(\Diamond_{#1}^{1})}
\newcommand{\spdiarii}[1]{(\Diamond_{#1}^{2})}

\newcommand{\nr}[1]{(n_{#1})}

\newcommand{\nspv}{\mathsf{NS}(\vocab)}

\newcommand{\prove}{\mathtt{Prove}_{\vocab}}
\newcommand{\proofsearch}{\mathtt{ProofSearch}_{\vocab}}
\newcommand{\true}{\mathtt{True}}
\newcommand{\false}{\mathtt{False}}
\newcommand{\scid}{id}
\newcommand{\sccon}{\land}
\newcommand{\scdis}{\lor}
\newcommand{\scspbox}{\spbox{s}}
\newcommand{\scspdia}{\spdia{s}}

\newcommand{\scns}{n_{s}}
\newcommand{\scastdia}{\spdia{\ast}}

\newcommand{\act}{\circ}
\newcommand{\inact}{\bullet}

\newcommand{\model}{\mathcal{M}}

\newcommand{\vocab}{\mathcal{V}}
\newcommand{\pset}{\mathcal{P}}
\newcommand{\spset}{\mathcal{S}}
\newcommand{\dst}{\ast}

\newcommand{\pred}[1]{\small\mbox{\tt #1}}

\newcommand{\spform}[1]{\mathsf{#1}}

\newcommand{\mathcom}[3]{ \newcommand{#1}[#2]{\mbox{$#3$}}}
\mathcom{\rimp}{0}{\ \leftarrow\ }            

\mathcom{\con}{0}{\ \wedge\ }                 
\mathcom{\dis}{0}{\ \vee\ }                   
\mathcom{\n}{0}{\neg}                     
\mathcom{\dimp}{0}{\ \leftrightarrow\ }       
\mathcom{\corresponds}{0}{\ \Lleftarrow\! \! \Rrightarrow\ }
\mathcom{\A}{0}{\forall}                  
\mathcom{\E}{0}{\exists}     

\def\Box{\mathop\square}
\def\Diamond{\mathop\lozenge}

\mathcom{\tuple}{1}{\langle #1 \rangle}

\def\textquote#1{`#1'}

\def\PropStandpointLogic{{\mathbb{S}}}

\def\s#1{\hbox{$\mathsf{#1}$}\xspace}
\def\star{\hbox{$*$}\xspace}

\def\standb#1{\Box\nolimits_{\spform{#1}}\xspace}
\def\standd#1{\Diamond\nolimits_{\spform{#1}}\xspace}

\def\allstandb{\standb{*}}
\def\allstandd{\standd{*}}

\def\modelclass{\mathfrak{M}_{\PropStandpointLogic}}
\def\model{\mathcal{M}}

\xspace
\def\f\xspacestandtopre{\hbox{$\sigma\,$}\xspace}
\def\fpretov\xspacealue{\hbox{$\delta\,$}\xspace}

\def\ModSat#1||-#2{#1\models #2}
\def\NotModSat#1||-#2{#1\nvDash #2}

\newtheorem{definition}{Definition}
\newtheorem{corollary}{Corollary}
\newtheorem{lemma}{Lemma}
\newtheorem{example}{Example}
\newtheorem{theorem}{Theorem}

\title{Automating Reasoning with Standpoint Logic via Nested Sequents}

\author{ }

\author{%
Tim S. Lyon \and Lucía Gómez Álvarez\\
\affiliations Computational Logic Group, Faculty of Computer Science, TU Dresden
\emails
\{timothy\_stephen.lyon, lucia.gomez\_alvarez\}@tu-dresden.de
}

\begin{document}

\maketitle

\begin{abstract}
Standpoint logic is a recently proposed formalism in the context of knowledge integration, which advocates a multi-perspective approach permitting reasoning with a selection of diverse and possibly conflicting standpoints rather than forcing their unification. In this paper, we introduce nested sequent calculi for propositional standpoint logics---proof systems that manipulate trees whose nodes are multisets of formulae---and show how to automate standpoint reasoning by means of non-deterministic proof-search algorithms. To obtain worst-case complexity-optimal proof-search, we introduce a novel technique in the context of nested sequents, referred to as \emph{coloring}, which consists of taking a formula as input, guessing a certain coloring of its subformulae, and then running proof-search in a nested sequent calculus on the colored input. Our technique lets us decide the validity of standpoint formulae in $\conp$ since proof-search only produces a \emph{partial} proof relative to each permitted coloring of the input. We show how all partial proofs can be fused together to construct a complete proof when the input is valid, and how certain partial proofs can be transformed into a counter-model when the input is invalid. These ``certificates'' (i.e. proofs and counter-models) serve as explanations of the (in)validity of the input. 
\end{abstract}

\section{Introduction}\label{sec:introduction}




\subsubsection{Standpoint Logic.} The fact that knowledge bases (KBs) encode the standpoints of their creators (e.g. in the form of viewpoints, contextual factors or semantic commitments) is the source of well-known challenges in the area of knowledge integration. Since semantic heterogeneity between the sources is to be expected, inconsistencies may arise if we attempt to combine them into a single conflict-free conceptual model. 
To illustrate this, consider three KBs: $\mathsf{C}$, a \textquote{common-sense} representation of colours; $\mathsf{H}$, a KB by a \emph{house painting} business that reuses and extends $\mathsf{C}$; and $\mathsf{R}$, a KB that formalizes the \emph{RYB color model}, from the fine arts tradition.

 
 \begin{example}\label{example:first-colours}
According to $\mathsf{C}$, basic colours such as $\pred{Blue}$ and $\pred{Green}$ are disjoint. $\mathsf{H}$ complies with $\mathsf{C}$ and further specifies that $\pred{Teal}$ is $\pred{Green}$.
In contrast, according to $\mathsf{R}$ it is unequivocal that $\pred{Teal}$ is both $\pred{Green}$ and $\pred{Blue}$. 
Generally, it is conceivable that $\pred{Blue}$ holds.
\end{example}

These sources cannot be merged without the undesired effect of inconsistency, and circumventing it requires either knowledge weakening or duplication 
\cite{Pesquita2013ToAlignments}. Instead, one may wish to jointly reason with the KBs, treating them as alternative standpoints on a domain.


\begin{figure}
\begin{center}
\begin{tikzpicture}

\node[
] (w1) [] {$\pred{Blue}$}; 

\node[
] (w2) [right of=w1, xshift=3em] {$\pred{Green}$};

\draw[thick] (w1) -- (w2) node [midway,above, yshift=-1pt] {$\perp$};

\node[draw, very thick, 
 color=black, rounded corners, inner xsep=5pt, inner ysep=6pt, fit=(w1) (w2)] (m1) {}; 
 
\draw node [below right, inner sep=2pt] at (m1.south west) {\footnotesize $\s{C}$};


\node[
] (w31) [below of=w1, yshift=-1.5em] {$\pred{Blue}$}; 

\node[
] (w32) [right of=w31, xshift=3em] {$\pred{Green}$};


\node[
] (w34) [below of=w32, yshift=0.5em] {$\pred{Teal}$};

\draw[thick] (w31) -- (w32) node [midway,above, yshift=-1pt] {$\perp$};

\draw[->,thick] (w34) -- (w32) node [] {};


\node[draw, very thick, 
 color=black, rounded corners, inner xsep=5pt, inner ysep=6pt, fit=(w31) (w32) 
 (w34)] (m3) {}; 
 
\draw node [below right, inner sep=2pt] (l3) at (m3.south west) {\footnotesize $\s{H}$};

\draw[thick] (m1) -- (m3) node [midway,right] {$\preceq$};

\node[
] (w5) [right of=w32, xshift=3em] {$\pred{Teal}$};


\node[
] (w8) [above of=w5,yshift=5pt] {$\pred{Green}$};

\node[
] (w7) [right of=w8, xshift=3em] {$\pred{Blue}$}; 


\draw[->,color=black,thick] (w5) -- (w7) node [] {};

\draw[->,color=black,thick] (w5) -- (w8) node [] {};



\node[draw, very thick, 
 color=black, rounded corners, inner xsep=5pt, inner ysep=6pt, fit=(w5) 
 (w7) (w8)] (m2) {}; 
 
\draw node [below right, inner sep=2pt] at (m2.south west) {\footnotesize $\s{R}$};

 

 \end{tikzpicture}
\end{center}
\vspace{-1em}
\caption{Diagrams of $\s{C}$, $\s{H}$ and $\s{R}$. $\preceq$ indicates that $\s{H}$ extends $\s{C}$. }
\label{fig:color-models}
\vspace{-1em}

\end{figure}

Standpoint logic \cite{gomez2021standpoint} is a simple multi-modal logic intended for the representation of knowledge relative to different, possibly conflicting, perspectives. The framework introduces the labeled modalities $\spbox{s}$ and $\spdia{s}$ for each standpoint $s$, 
where $\spbox{s} \phi$ is read as ``according to $s$, it is \emph{unequivocal} that $\phi$'' and $\spdia{s} \phi$ as ``according to $s$, it is \emph{conceivable} that $\phi$''. In addition, $s\preceq s'$ indicates that the standpoint $s$ is \emph{sharper} than 
$s'$, that is, $s$ complies with $s'$ and further specifies it.



\begin{enumerate}[series=fexample,label={\rm(F\arabic*)},ref={\rm (F\arabic*)},leftmargin=2.2em,labelwidth=2em] 
    \item	$\standb{C}\neg(\pred{Blue}\!\con\!\pred{Green}) \con  \standb{R}(\pred{Teal}\! \rightarrow\! (\pred{Blue}\!\con\!\pred{Green})) $\label{formula:ex_CandH}
\end{enumerate}
\begin{enumerate*}[resume*=fexample,label={\rm(F\arabic*)},ref={\rm (F\arabic*)},labelwidth=2em,itemjoin={\quad}] 
	\item $(\s{H}\!\preceq\!\s{C})$\label{formula:ex_prec}
	\item $\standb{H}(\pred{Teal}\!\rightarrow\!\pred{Green})$\label{formula:ex_J}
	\item $\allstandd(\pred{Blue})$\label{formula:ex_universal}
\end{enumerate*}
\vspace{0.3em}

The formulae \ref{formula:ex_CandH}-\ref{formula:ex_universal} formalize \Cref{example:first-colours}, (illustrated in \Cref{fig:color-models}) in propositional standpoint logic. \ref{formula:ex_CandH} encodes that $\pred{Blue}$ and $\pred{Green}$ are unequivocally disjoint according to standpoint $\mathsf{C}$, while according to standpoint $\mathsf{R}$ it is unequivocal that $\pred{Teal}$ implies both $\pred{Blue}$ and $\pred{Green}$. \ref{formula:ex_prec} encodes that $\s{H}$ includes the knowledge of $\s{C}$, and \ref{formula:ex_J} that $\pred{Teal}$ is $\pred{Green}$ according to $\s{H}$. Last, \ref{formula:ex_universal} encodes that $\pred{Blue}$ holds under some interpretations by using the universal standpoint $\star$, which sits atop any hierarchy of standpoints and is used to reference knowledge that is unequivocally true or conceivable among all perspectives.


In addition to representing unequivocal and conceivable facts (e.g. \ref{formula:ex_J} and \ref{formula:ex_universal}), which may be relative to standpoints (e.g. \ref{formula:ex_CandH}), hold universally (e.g. \ref{formula:ex_universal}), or establish a hierarchy of standpoints (e.g. \ref{formula:ex_prec}), one may also express (in)determinate knowledge by means of the (definable) dual operators $\mathcal{I}_{s}$ and $\mathcal{D}_{s}$. The indeterminacy operator $\mathcal{I}_{s}\phi:= \spdia{s}\phi \con \spdia{s}\neg\phi\ $ makes explicit that both $\phi$ and $\neg\phi$ are conceivable in the context of $s$, thus making $\phi$ inherently indeterminate. Finally, the framework can be used to establish correspondences or bridges between the standpoints themselves. For instance, \ref{formula:ex_bridge} encodes that if something is $\pred{Teal}$ according to $\mathsf{R}$, then it is $\pred{Green}$ for $\mathsf{C}$ and $\mathsf{R}$. 

\begin{enumerate}[resume*=fexample,label={\rm(F\arabic*)},ref={\rm (F\arabic*)},leftmargin=2.2em,labelwidth=2em] 
  \item	$\standb{R}(\pred{Teal})\rightarrow(\standb{C}\pred{Green} \con \standb{R}\pred{Green})$\label{formula:ex_bridge}
\end{enumerate}

Natural reasoning tasks over multi-standpoint specifications include gathering unequivocal or undisputed knowledge, determining knowledge that is relative to a standpoint or a set of them, and contrasting the knowledge that can be inferred from different standpoints.
To illustrate, let us assume $\allstandb\pred{Teal}$ and examine some inferences that we can draw from this in the setting of \Cref{example:first-colours}. 
On the one hand, from \ref{formula:ex_bridge}, \ref{formula:ex_J}, and $\allstandb\pred{Teal}$ we obtain that green is unequivocal for the three standpoints: $\standb{C}\pred{Green}$, $\standb{H}\pred{Green}$ and $\standb{R}\pred{Green}$.
On the other hand, we can infer the global indeterminacy of blue $\mathcal{I}_{\star}\pred{Blue}$, because (i) $\pred{Teal}$ holds universally, (ii) it is unequivocal for $\s{R}$ that $\pred{Teal}$ implies $\pred{Blue}$ \ref{formula:ex_CandH}, hence $\allstandd\pred{Blue}$, and (iii) we know $\standb{C}\pred{Green}$, which together with \ref{formula:ex_CandH} implies $\standb{C}\neg\pred{Blue}$ and thus $\allstandd\neg\pred{Blue}$. 

Conveniently, the satisfiability problem in propositional standpoint logic is known to be $\np$-complete \cite{gomez2021standpoint}, in pleasant contrast to the $\pspace$-completeness normally exhibited by multi-modal epistemic logics, such as the closely related $\mathsf{KD45}_{n}$.\footnote{Standpoint logic introduces sharpenings and stronger interaction axioms than $\mathsf{KD45}_{n}$, as discussed in \cite{GomezAlvarezThesis}.} 
 This result, obtained via a translation to one-variable first-order logic, makes the framework attractive in applied scenarios, and prompts our work to provide a suitable proof-theory for standpoint logic. Not only can our proof systems be leveraged to provide a proof-search procedure deciding the validity of standpoint formulae, but our proof-theoretic approach yields \emph{witnesses}, that is, proofs of valid formulae and counter-models of invalid formulae. Such ``certificates'' (i.e. proofs and counter-models) possess explanatory value, and may be used, for instance, to trace the standpoints involved in a certain inference; e.g. when a global indeterminacy such as $\mathcal{I}_{\star}\pred{Blue}$ is inferred from a large collection of standpoints, we may want to gather the standpoints that hold contrasting views (in this case $\s{R}$ and $\{\s{H},\s{C}\}$, which can be easily extracted from a proof).
Thus, our reliance on proof theory provides essential information that may be used to answer ``why'' a certain piece of information holds while still allowing ``low'' complexity reasoning.

\subsubsection{Nested Sequents and Proof Theory.} Since their inception, sequent systems---which consist of inference rules that syntactically manipulate pairs of multisets of formulae---have proven themselves fruitful in writing decision algorithms for logics~\cite{Dyc92,Gen35a,Gen35b,Sla97}. A crucial feature of such systems, and their use in decidability, is the so-called \emph{subformula property}, which a sequent system 
has \iffi the premise(s) of each inference rule only contain subformulae of the conclusion of the rule. (NB. Systems with the subformula property are also referred to as \emph{analytic}.) 
 With the goal of securing this property for proof systems for theories \emph{beyond classical propositional logic} (e.g. the modal logics $\mathsf{Kt}$ and $\mathsf{S5}$), more sophisticated sequent systems extending Gentzen's original formalism were eventually proposed; e.g., see~\cite{Bel82,Sim94,Wan02}. In this paper, we employ one such extended formalism, viz. the \emph{nested sequent formalism}~\cite{Bru09,Bul92,Kas94,Pog09}, which utilizes trees of multisets of formulae in deriving theorems. Such systems have proven well-suited for automated reasoning with modal and related logics, being used (for instance) in the writing of decision/proof-search algorithms~\cite{Bru09,TiuIanGor12} and the extraction of interpolants~\cite{FitKuz15,LyoTiuGorClo20}.

Drawing on ideas from the \emph{structural refinement} methodology, detailed in~\cite{Lyo21thesis} and used to provide nested sequent systems for diverse classes of modal and constructive logics (see~\cite{LyoBer19,Lyo21a,Lyo21b}), our first contribution in this paper is the introduction of analytic nested sequent systems (each dubbed $\nspv$ with $\vocab$ a certain parameter) for propositional standpoint logics~\cite{gomez2021standpoint}. For our second contribution, we exploit our nested systems to write concrete, \emph{worst-case complexity-optimal} proof-search algorithms (deciding the validity of propositional standpoint formulae in $\conp$), which apply inference rules from $\nspv$ \emph{in reverse} on an input formula with the goal of building a proof thereof. Whereas typical proof-search algorithms operate deterministically and attempt to build a \emph{complete} proof of the input, we introduce a novel technique (our third contribution) referred to as \emph{coloring}, which performs proof-search \emph{non-deterministically} and which only constructs a \emph{partial} proof of the input relative to each non-deterministic choice. The technique of coloring involves first guessing a particular labeling of the subformulae of an input formula with \emph{active} $\act$ and \emph{inactive} $\inact$ labels, with the proof-search algorithm subsequently only processing data deemed active. An interesting consequence of this technique is 
 the attainment of a $\conp$ proof-search algorithm as the partial proofs constructed during proof-search are at most polynomially larger than the input and only require polynomial time to compute. Moreover, in the instance where the input formula is invalid, we show how to construct a counter-model from failed proof-search, and in the instance where our input formula is valid, we provide a procedure that generates a \emph{complete} proof witnessing the validity of the input formula by patching together all \emph{partial} proofs (our fourth contribution).

\subsubsection{Organization of Paper.} Our paper is organized as follows: \sect~\ref{sec:standpoint-logic} presents the syntax and semantics of propositional standpoint logic. In \sect~\ref{sec:nested-calculi}, we introduce our nested sequent systems for propositional standpoint logics, proving such systems sound and concluding their completeness. In the penultimate section (\sect~\ref{sec:proof-search}), we introduce the method of coloring and show how to automate reasoning with standpoint logics, that is, we provide a (worst-case complexity-optimal) proof-search algorithm deciding the validity of propositional standpoint formulae in $\conp$. The final section (\sect~\ref{sec:conclusion}) concludes the paper and discusses future work.

\section{Standpoint Logic}\label{sec:standpoint-logic}

Let us now specify the syntax of propositional standpoint logic (SL), denoted by $\PropStandpointLogic$.

\begin{definition}[Syntax of Standpoint Logic]\label{def:logical-languages} Let $\vocab = \langle \pset, \spset \rangle$ be a \emph{vocabulary} where $\pset$ is a non-empty set of propositional variables and $\spset$ is a set of standpoint symbols containing the distinguished symbol $\dst$, i.e. the \emph{universal standpoint}. We define the language $\langsharp := \{s \preceq s' \ | \ s,s' \in \spset\}$, and refer to formulae in $\langsharp$ as \emph{sharpening statements}. The language $\langv$ is defined via the following grammar in BNF:
$$
\phi ::=  p \ | \ \neg p \ |\ (\phi \lor \phi) \ | \ (\phi \land \phi) \ | \ \spbox{s} \phi \ | \ \spdia{s} \phi
$$
where $p \in \pset$ and $s \in \spset$. We also use $\top$ and $\bot$ as shorthands with the usual definitions.

Last, for $\Gamma \subseteq \langsharp$ and $\phi \in \langv$, we define a \emph{standpoint implication} to be a formula of the form $\bigwedge \Gamma \rightarrow \phi$, where $\bigwedge \Gamma$ is a conjunction of all elements of $\Gamma$, which equals $\top$ when $\Gamma$ is empty.
\end{definition}

 We make use of formulae in negation normal form as this will simplify the structures present in our nested systems and enhance the readability of our proof theory. To further simplify, we also assume w.l.o.g. that sets of sharpening statements are (1) free of cycles $s_{1} \preceq s_{2}, \ldots, s_{n} \preceq s_{1}$ and (2) omit occurrences of $\star$. Assumption (1) is permitted since any standpoint implication containing a cycle $s_{1} \preceq s_{2}, \ldots, s_{n} \preceq s_{1}$ of standpoints is equivalent to one where the cycle is deleted and all occurrences of $s_{1}, \ldots, s_{n-1}$ are replaced by $s_{n}$ in the formula. Regarding assumption (2), any sharpening statement with $\star$ is either of the form $s\preceq\star$, and is thus trivial (see \dfn~\ref{def:semantic-clauses} below), or is of the form $\star\preceq s$, in which case $s$ can be systematically replaced by $\star$ in a standpoint implication to obtain an equivalent one.

\begin{definition}[Subformula and Size]\label{DEF:subformulas}
 We define the set of \emph{subformulae} of $\phi$, denoted $\sufo{\phi}$, recursively as follows:
\begin{itemize}

\item $\sufo{p} := \{p\}$ and $\sufo{\neg p} :=  \{\neg p\}$;

\item $\sufo{\heartsuit \psi} := \{\heartsuit \psi\} \cup \sufo{\psi}$;

\item $\sufo{\psi \otimes \chi} := \{\psi \otimes \chi\} \cup \sufo{\psi} \cup \sufo{\chi}$.

\end{itemize}
with $\heartsuit \in \{\spdia{s}, \spbox{s} \ | \ s \in \spset\}$ and $\otimes \in \{\lor, \land\}$. We say that $\psi$ is a \emph{subformula} of $\phi$ \iffi $\psi \in \sufo{\phi}$, and define the \emph{size} of a formula $\phi$ in $\langv$, denoted $|\phi|$, to be equal to $|\sufo{\phi}|$, i.e. to the number of its subformulae.
\end{definition}

In what follows, we introduce the semantics of SL, defined over a structure of precisifications, which is akin to the usual structure of possible worlds. A \emph{precisification} is a complete and consistent way in which the state of affairs can be described with a given vocabulary, and standpoints are modeled as sets of precisifications considered admissible. This strategy of modelling the variability of natural language as hyper-ambiguity is based on the theory supervaluationism \cite{Fine1975,Keefe1997TheoriesVagueness}, which standpoint logic draws from \cite{Alvarez2018DealingSemantics,Alvarez2017TalkingTerms}.

\begin{definition}[Standpoint Model]\label{def:standpoint-models} Given a vocabulary $\vocab$, a model $\model$ (over $\vocab$) is a triple $\langle \Pi, \sigma, \delta \rangle$, where $\Pi$ is a non-empty set of precisifications, $\sigma : \spset \to 2^{\Pi}$, and $\delta : \pset \to 2^{\Pi}$ with $\sigma(s) \neq \emptyset$ for all $s \in \spset$ and 
$\sigma(\dst) = \Pi$. The set of all such models is denoted by $\modelclass$. 
\end{definition}

\begin{definition}[Semantic Clauses]\label{def:semantic-clauses} Let $\Gamma \subseteq \langsharp$ and $\phi, \psi \in \langv$. Moreover, let $\model = \langle \Pi, \sigma, \delta \rangle$ be a standpoint model with $\pi \in \Pi$. We recursively define the satisfaction of a formula on $\model$ at $\pi$ accordingly:
\begin{itemize}

\item $\model, \pi \models p$ \ifandonlyif  $\pi \in \delta(p)$;

\item $\model, \pi \models \neg p$ \ifandonlyif $\pi \not\in \delta(p)$;

\item $\model, \pi \models \phi \land \psi$ \ifandonlyif $\model, \pi \models \phi$ and $\model, \pi \models \psi$;

\item $\model, \pi \models \phi \lor \psi$ \ifandonlyif $\model, \pi \models \phi$ or $\model, \pi \models \psi$;

\item $\model, \pi \models \spdia{s} \phi$ \ifandonlyif for some $\pi' \in \sigma(s)$, $\model, \pi' \models \phi$;

\item $\model, \pi \models \spbox{s} \phi$ \ifandonlyif for all $\pi' \in \sigma(s)$, $\model, \pi' \models \phi$;

\item $\model, \pi \models s \sharper s'$ \ifandonlyif $\sigma(s) \subseteq \sigma(s')$;


\item $\model, \pi \models \bigwedge \Gamma$ \ifandonlyif $\model, \pi \models s \sharper s'$ for all $s \sharper s' \in \Gamma$;

\item $\model, \pi \models \bigwedge \Gamma \rightarrow \phi$ \ifandonlyif  $\model, \pi \models \bigwedge \Gamma$ implies $\model, \pi \models \phi$;

\item $\model \models \bigwedge \Gamma \rightarrow \phi$ \ifandonlyif $\model, \pi \models \bigwedge \Gamma \rightarrow \phi$ for all $\pi \in \Pi$.


\end{itemize}

A standpoint implication $\bigwedge \Gamma \rightarrow \phi$ is defined to be \emph{valid} (relative to a vocabulary $\vocab$) \iffi it is true on each model $\model \in \modelclass$; it is defined to be \emph{invalid} (relative to $\vocab$)  otherwise. 

For a vocabulary $\vocab$, the \emph{standpoint logic} $\spv$ is the set of all valid standpoint implications $\bigwedge \Gamma \rightarrow \phi$ over $\modelclass$.
\end{definition}

It is worth remarking that the specification of sharpening statements in a separate language (viz. $\langsharp$) and the above definition of satisfiability and validity contrast with the original presentation in \cite{gomez2021standpoint}. However, this specification simplifies our treatment of sharpening statements, which previously served as atomic propositions in the language $\langv$. In fact, these statements are obsolete in extensions of the language allowing set theoretical combinations of standpoints in modalities (which is the object of current research). 
 Moreover, in these extensions, the natural requirement of inner consistency (i.e. the non-emptiness of $\sigma(s)$, for each $s \in \spset$) of standpoints is relaxed, which can be easily reflected in our nested sequent systems by dropping the $\nr{s}$ rule (see \fig~\ref{fig:calculus} in \Cref{sec:nested-calculi}).


\section{Nested Sequent Systems}\label{sec:nested-calculi}

We define a \emph{nested sequent} (which we will also refer to as a \emph{sequent}) to be a formula of the form $\sprset \sep \nseq$ with $\sprset$ and $\nseq$ defined via the following grammars in BNF:
$$
\sprset ::= s \preceq s' \ | \ \emptyset \ | \ \sprset, \sprset \quad \nseq ::= \Sigma \ | \ \nseq, (s)[\Sigma]_{\pi}
$$
$$
\Sigma ::= \phi \ | \ \emptyset \ | \ \Sigma, \Sigma
$$
where $s, s' \in \spset \setminus \{\dst\}$, $\phi \in \langv$, and $\pi$ is among a countably infinite set of labels $\{\pi_{i} \ | \ i \in \mathbb{N} \setminus \{0\}\}$. We use $\ns$ and $\nsii$ (occasionally annotated) to denote nested sequents and note that we employ the use of labels as this proves useful in extracting a counter-model from failed proof-search (see \thm~\ref{thm:correctness}). Moreover, each nested sequent $\sprset \sep \nseq$ with $\nseq = \Sigma_{0}, (s_{1})[\Sigma_{1}]_{\pi_{1}}, \ldots, (s_{n})[\Sigma_{n}]_{\pi_{n}}$ possesses a special structure; namely, the \emph{antecedent} $\sprset$ is a set of sharpening statements of the form $s \preceq s'$, 
and the \emph{consequent} $\nseq$ is a multiset  
 encoding a tree of depth 1 whose nodes are multisets of formulae from $\langv$. The consequent $\nseq$ can be expressed graphically as follows:
\begin{center}
\begin{tabular}{c c c}
\xymatrix@C=1em{
		&  & & \Sigma_{0}\ar@{->}[dlll]|-{s_{1}}\ar@{->}[dl]|-{s_{2}}\ar@{->}[dr]|-{s_{n-1}}\ar@{->}[drrr]|-{s_{n}} & & &   		\\
	\Sigma_{1} & & \Sigma_{2}	& \hdots  & \Sigma_{n-1} & & \Sigma_{n} 
}
\end{tabular}
\end{center}
We refer to a multiset $\Sigma_{i}$ occurring in the consequent of a nested sequent as a \emph{component}, and note that components (along with the antecedent and consequent) are permitted to be empty $\emptyset$. Intuitively, components correspond to precisifications in a standpoint model. It is also worthwhile to define the relation $\sprel$ on standpoints as this will be used as a side condition dictating applications of certain inference rules:

\begin{definition}\label{def:ref-tra-sharpening-closure} For a nested sequent $\sprset \sar \nseq$, let $\sprel \mathop{\subseteq} \spset \times \spset$ be the minimal reflexive and transitive relation such that
\begin{itemize}
    \item $s \sprel s'$ for every $s \preceq s' \in \Gamma$, and
    \item $s \sprel *$ for every $s \in \mathcal{S}$.
\end{itemize}
\end{definition}



A nice feature of nested sequents is that such objects typically permit a formula translation,  e.g.~\cite{Bru09,Bul92,Kas94,Pog09}, meaning that our logical semantics can be lifted to the language of our proof systems without introducing an extended semantics for nested sequents.

\begin{definition}[Formula Interpretation]\label{def:formula-interpretation} We define the \emph{formula interpretation} of a nested sequent $\sprset \sep \nseq$ with $\nseq = \Sigma_{0}, (s_{1})[\Sigma_{1}]_{\pi_{1}}, \ldots, (s_{n})[\Sigma_{n}]_{\pi_{n}}$ as follows:
$$
\iota(\sprset \sep \nseq) := \bigwedge \sprset \rightarrow \bigvee \Sigma_{0} \lor \bigvee_{1 \leq i \leq n} \spbox{s_{i}}(\bigvee \Sigma_{i}) 
$$
We define $\sprset \sep \nseq$ to be valid \ifandonlyif $\iota(\sprset \sep \nseq)$ is valid. Also, we note that $\bigwedge \emptyset = \top$ and $\bigvee \emptyset = \bot$, as usual.
\end{definition}

A uniform presentation of our nested calculi is given in \fig~\ref{fig:calculus}. We let $\nspv$ denote the corresponding nested sequent calculus over a vocabulary $\vocab$. Our inference rules make use of the brackets `$\{$' and `$\}$' in the consequent of a nested sequent to indicate that the displayed formula(e) occur in some component. In particular, given a nested sequent $\sprset \sep \nseq$, where $\nseq$ is of the form $\Sigma_{0}, (s_{1})[\Sigma_{1}]_{\pi_{1}}, \ldots, (s_{i})[\Sigma_{i}]_{\pi_{i}}, \ldots, (s_{n})[\Sigma_{n}]_{\pi_{n}}$, the notation $\sprset \sep \nseq\{\phi\}_{\pi_{i}}$ indicates that $\phi$ occurs in $\Sigma_{i}$; additionally, we use $\sprset \sep \nseq\{\phi\}_{\pi_{0}}$ to indicate that $\phi$ occurs in $\Sigma_{0}$, i.e. the label $\pi_{0}$ is used to reference the multiset $\Sigma_{0}$ serving as the root of the tree encoded by the consequent.

\begin{figure*}[t]

\begin{center}
\begin{tabular}{c c c}

\AxiomC{}
\RightLabel{$\id$}
\UnaryInfC{$\sprset \sep \nseq\{p, \neg p\}_{\pi}$}
\DisplayProof

&

\AxiomC{$\sprset \sep \nseq\{\phi, \psi\}_{\pi}$}
\RightLabel{$\disr$}
\UnaryInfC{$\sprset \sep \nseq\{\phi \lor \psi\}_{\pi}$}
\DisplayProof

&

\AxiomC{$\sprset \sep \nseq\{\phi\}_{\pi}$}
\AxiomC{$\sprset \sep \nseq\{\psi\}_{\pi}$}
\RightLabel{$\conr$}
\BinaryInfC{$\sprset \sep \nseq\{\phi \land \psi\}_{\pi}$}
\DisplayProof
\end{tabular}
\end{center}

\begin{center}
\begin{tabular}{c c}
\AxiomC{$\sprset \sep \nseq\{\spbox{s} \phi\}_{\pi}, (s)[\phi]_{\pi'}$}
\RightLabel{$\spboxr{s}^{\dag_{1}}$}
\UnaryInfC{$\sprset \sep \nseq\{\spbox{s} \phi\}_{\pi}$}
\DisplayProof

&

\AxiomC{$\sprset \sep \nseq, (s)[\emptyset]_{\pi'}$}
\RightLabel{$\nr{s}^{\dag_{1}}$}
\UnaryInfC{$\sprset \sep \nseq$}
\DisplayProof
\end{tabular}
\end{center}

\begin{center}
\begin{tabular}{c c c}
\AxiomC{$\sprset \sep \nseq\{\spdia{s} \phi\}_{\pi}, (s')[\Sigma, \phi]_{\pi'}$}
\RightLabel{$\spdiari{s}^{\dag_{2}}$}
\UnaryInfC{$\sprset \sep \nseq\{\spdia{s} \phi\}_{\pi}, (s')[\Sigma]_{\pi'}$}
\DisplayProof

&

\AxiomC{$\sprset \sep \nseq, (s')[\spdia{s} \phi, \phi, \Sigma]_{\pi}$}
\RightLabel{$\spdiarii{s}^{\dag_{2}}$}
\UnaryInfC{$\sprset \sep \nseq, (s')[\spdia{s} \phi, \Sigma]_{\pi}$}
\DisplayProof

&

\AxiomC{$\sprset \sep \phi, \nseq\{\spdia{*} \phi\}_{\pi}$}
\RightLabel{$\spdiar{*}$}
\UnaryInfC{$\sprset \sep \nseq\{\spdia{*} \phi\}_{\pi}$}
\DisplayProof
\end{tabular}
\end{center}

\caption{The nested calculus $\nspv$ with $\vocab = \langle \pset, \spset \rangle$ a vocabulary. We note that $\pi$ is permitted to be any label from $\{\pi_{i} \ | \ i \in \mathbb{N} \setminus \{0\}\}$ and that $\nspv$ contains a copy of $\spboxr{s}$, $\nr{s}$, $\spdiari{s}$, and $\spdiarii{s}$ for each $s \in \spset$. The side condition $\dag_{1}$ stipulates that the rule is applicable only if the label $\pi'$ is fresh and $\dag_{2}$ stipulates that the rule is applicable only if $s' \preceq^{*}_{\sprset} s$.}\label{fig:calculus}

\end{figure*}

To make the functionality of each rule in $\nspv$ precise, we explicitly state the operation performed by each rule. With the exception of the premise-free $\id$ rule, we explain for each rule how the \emph{premise(s)} (the nested sequent(s) occurring above the horizontal inference line) are obtained from the \emph{conclusion} (the nested sequent occurring below the horizontal inference line). This explanation is consistent with how the rules are applied (bottom-up) during proof-search as described in the following section. Also, in accordance with standard proof-theoretic terminology~\cite{Bus98,Tak13}, we refer to the formula that is explicitly displayed in the conclusion of a rule as \emph{principal}, and indicate the principal formulae in our explanation of the rules below to make this precise for the reader.
\begin{description}

\item[$\id$] A nested sequent is \emph{initial}, and may be used to begin a derivation, so long as some component 
 contains both $p$ and $\neg p$ (the principal formulae);

\item[$\disr$] If a component $\Sigma_{i}$ of the conclusion contains $\phi \lor \psi$ (the principal formula), then adding $\phi$ and $\psi$ to $\Sigma_{i}$ yields the premise;

\item[$\conr$] If a component $\Sigma_{i}$ of the conclusion contains $\phi \land \psi$ (the principal formula), then adding $\phi$ to $\Sigma_{i}$ yields the left premise and adding $\psi$ to $\Sigma_{i}$ yields the right premise;

\item[$\spboxr{s}$] For any $s \in \spset$, if a component $\Sigma_{i}$ of the conclusion contains $\spbox{s} \phi$ (the principal formula), then appending the consequent $\nseq$ with $(s)[\phi]_{\pi'}$, where $\pi'$ is fresh (i.e. it does not occur in the conclusion), yields the premise;

\item[$\nr{s}$] For any $s \in \spset$, we may append the consequent of the conclusion with $(s)[\emptyset]_{\pi'}$ to obtain the premise so long as $\pi'$ is fresh;

\item[$\spdiari{s}$] For any $s \in \spset$, if a component $\Sigma_{i}$ of the conclusion contains $\spdia{s} \phi$ (the principal formula), the consequent contains a nesting $(s')[\Sigma]_{\pi'}$, and $s' \preceq^{*}_{\sprset} s$, then adding $\phi$ to the nesting $(s')[\Sigma]_{\pi'}$ yields the premise;

\item[$\spdiarii{s}$] For any $s \in \spset$, if the consequent contains a nesting of the form $(s')[\spdia{s} \phi, \Sigma]_{\pi'}$ with $\spdia{s} \phi$ the principal formula, and $s' \preceq^{*}_{\sprset} s$, then adding $\phi$ to the nesting $(s')[\spdia{s} \phi, \Sigma]_{\pi'}$ yields the premise;

\item[$\spdiar{\dst}$] If a component $\Sigma_{i}$ of the conclusion contains $\spdia{\dst} \phi$ (the principal formula), then prepending the consequent $\nseq$ with $\phi$ (i.e. adding $\phi$ to $\Sigma_{0}$) yields the premise.
\end{description}

\vspace{0.3em}

\begin{example} Below, we provide an example of a nested sequent derivation. To minimize the width of the proof, we let $\phi$ denote $\spdia{s} \spdia{\dst} \neg p \lor \spbox{s'} p$.
\begin{center}
\AxiomC{}
\RightLabel{$\id$}
\UnaryInfC{$s' \preceq s \sar \phi, \spdia{s} \spdia{\dst} \neg p, \spbox{s'} p, (s')[\spdia{\dst} \neg p, \neg p, p]$}
\RightLabel{$\spdiarii{\dst}$}
\UnaryInfC{$s' \preceq s \sar \phi,  \spdia{s} \spdia{\dst} \neg p, \spbox{s'} p, (s')[\spdia{\dst} \neg p, p]$}
\RightLabel{$\spdiari{s}$}
\UnaryInfC{$s' \preceq s \sar \phi,  \spdia{s} \spdia{\dst} \neg p, \spbox{s'} p, (s')[p]$}
\RightLabel{$\spboxr{s'}$}
\UnaryInfC{$s' \preceq s \sar \phi,  \spdia{s} \spdia{\dst} \neg p, \spbox{s'} p$}
\RightLabel{$\disr$}
\UnaryInfC{$s' \preceq s \sar \phi$}
\DisplayProof
\end{center}
Observe that $\spdiari{s}$ is applicable as $s' \sprel s$ holds due to the antecedent, and $\spdiarii{\dst}$ is applicable as $s' \sprel \dst$ holds by definition (see \dfn~\ref{def:ref-tra-sharpening-closure}).
\end{example}


We now prove that our calculi are sound (\thm~\ref{thm:soundness}, building on \lem~\ref{lem:sprel-implies-sharpening}), that is, that every nested sequent derivable in $\nspv$ is valid. We then state our completeness theorem, which is a consequence of the work in \sect~\ref{sec:proof-search}.

\begin{lemma}\label{lem:sprel-implies-sharpening}
Let $\sprset \vdash \nseq$ be a sequent, $\model = \langle \Pi, \sigma, \delta \rangle$ be a model with $\pi \in \Pi$, and $s, s' \in \spset$. If $\model, \pi \models \bigwedge \sprset$ and $s' \sprel s$, then $\model, \pi \models s' \preceq s$.
\end{lemma}

\begin{proof} Assume that $\model, \pi \models \bigwedge \sprset$ and $s' \sprel s$ for some $s', s \in \spset$. There are four cases to consider: 
\begin{enumerate}[label=(\arabic*),leftmargin=0.4 em,labelwidth=-1.2em]
    \item $s = \star$. The result is immediate as $\sigma(\dst) = \Pi$, and therefore, $\model, \pi \models s' \preceq \dst$ for every $s' \in \spset$ by \dfn~\ref{def:semantic-clauses}.
    \item $s' \preceq s \in \sprset$. From the assumption that $\model, \pi \models \bigwedge \sprset$ it follows that $\model, \pi \models s' \preceq s$.
    \item $s = s'$. Then, it is trivially implied that $\model, \pi \models s' \preceq s$ since $\sigma(s') = \sigma(s) \subseteq \sigma(s)$ by \dfn~\ref{def:semantic-clauses}.
    \item There are some $s_1, \ldots, s_n \in \spset$ such that $s' \preceq s_1,   \in \sprset$, $s_i \preceq s_{i+1}  \in \sprset$ for every $1 \leq i \leq n-1$, and $s_n \preceq s \in \sprset$, that is, $s' \sprel s$ is obtained by transitivity on a path in $\spset$. From this, together with the assumption that $\model, \pi \models \bigwedge \sprset$, it directly follows that $\model, \pi \models s' \preceq s$ by \dfn~\ref{def:semantic-clauses}. \qedhere
\end{enumerate}
\end{proof}

\begin{theorem}[Soundness]\label{thm:soundness}
If $\sprset \sep \nseq$ is derivable in $\nspv$, then $\sprset \sep \nseq$ is valid.
\end{theorem}

\begin{proof} 
 We prove the result by induction on the number of inferences in a given derivation, and assume that $\nseq$ is of the form $\Sigma_{0}, (s_{1})[\Sigma_{1}], \ldots, (s_{n})[\Sigma_{n}]$.

\textit{Base case.} In the base case, our derivation consists of a single application of the $\id$ rule. Hence, 
$$
\iota(\sprset \sep \nseq\{p,\neg p\}_{\pi_{i}}) := \bigwedge \sprset \rightarrow \bigvee \Sigma_{0} \lor \bigvee_{1\leq i \leq n} \spbox{s_{i}}(\bigvee \Sigma_{i}) 
$$
where $p, \neg p \in \Sigma_{i}$, for some $0\leq i \leq n$. Regardless, the consequent of the implication above will be satisfied in any model $\model$, implying that the above implication is valid.

\textit{Inductive step.} We make a case distinction based on the last rule applied, and show that if the conclusion of the rule is invalid, then at least one of the premises of the rule is invalid, that is to say, we show by contraposition that if the premise(s) is (are) valid, then the conclusion is valid. We only show the $\spboxr{s}$ and $\spdiari{s}$ cases as the remaining cases are simple or argued in a similar manner.

$\spboxr{s}$. We assume that $\Sigma_{1}$ is of the form $\spbox{s} \phi, \Sigma_{1}'$ with $\spbox{s} \phi$ principal; all remaining cases are similar. Furthermore, let us suppose that $\iota(\Gamma \vdash \Delta\{\spbox{s} \phi\}_{\pi_{1}}) :=$
$$\bigwedge \sprset \rightarrow \bigvee \Sigma_{0} \lor \spbox{s_{1}}(\spbox{s} \phi \lor \bigvee \Sigma_{1}')  \lor \bigvee_{2\leq i \leq n} \spbox{s_{i}}(\bigvee \Sigma_{i}) 
$$
is invalid. Then, $\model, \pi \not\models \spbox{s_{1}}(\spbox{s} \phi \lor \bigvee \Sigma_{1}')$ for some standpoint model $\model := \langle \Pi, \sigma, \delta \rangle$ with a precisification $\pi$. Hence, there exists a precisification $\pi' \in \sigma(s_{1})$ such that $\model, \pi' \not\models \spbox{s} \phi$, implying that there exists a precisification $\pi'' \in \sigma(s)$ such that $\model, \pi'' \not\models \phi$. It thus follows that $\model, \pi \not\models \spbox{s} \phi$, showing that the premise of $\spboxr{s}$ is invalid.

$\spdiari{s}$. Suppose that $\Sigma_{1}$ is of the form $\spdia{s} \phi, \Sigma_{1}'$ with $\spdia{s} \phi$ principal; all remaining cases are argued in a similar fashion. Assume that $s' \sprel s$ holds and that the following is invalid:
\begin{gather*}
  \iota(\sprset \sep \nseq\{\spdia{s} \phi\}_{\pi}, (s')[\Sigma]_{\pi'}) := \bigwedge  \sprset \rightarrow \bigvee \Sigma_{0} \lor \\
    \spbox{s_{1}}(\spdia{s} \phi \lor \bigvee \Sigma_{1}') \lor
  \Big(  \bigvee_{2\leq i \leq n} \!\!\!\spbox{s_{i}}\big(\bigvee \Sigma_{i}\big)\Big) \lor \spbox{s'}\big(\bigvee \Sigma\big)
\end{gather*}
Therefore, there is a standpoint model $\model := \langle \Pi, \sigma, \delta \rangle$ with $\pi\in\Pi$ such that $\model, \pi \not\models \spbox{s'}(\bigvee \Sigma)$ and $\model, \pi \not\models \spbox{s_{1}}(\spdia{s} \phi \lor \bigvee \Sigma_{1}')$, and such that $\sigma(s') \subseteq \sigma(s)$, by $s' \sprel s$ and \lem~\ref{lem:sprel-implies-sharpening}. 
This entails that there exists a precisification $\pi' \in \sigma(s')$ such that $\model, \pi' \not\models \bigvee \Sigma$, and that there exists a precisification $\pi_{1} \in \sigma(s_{1})$ such that $\model, \pi_{1} \not\models \spdia{s} \phi$. The latter further implies that for every precisification in $\sigma(s)$, and thus for $\pi'$ (since $\pi'\in\sigma(s') \subseteq \sigma(s)$), that $\model, \pi' \not\models \phi$. 
 Thus, the premise has been shown invalid.
\end{proof}

\begin{theorem}[Completeness]
If $\sprset \sep \nseq$ is valid, then $\sprset \sep \nseq$ is provable in $\nspv$.
\end{theorem}

\begin{proof}
The theorem follows from the correct (\thm~\ref{thm:correctness}) and terminating (\thm~\ref{thm:termination}) proof-search procedure given in the subsequent section (\sect~\ref{sec:proof-search}).
\end{proof}

\section{Automating Standpoint Logic via Proof-Search}\label{sec:proof-search}

We now employ our nested calculi in an algorithm that decides the validity of formulae for propositional standpoint logics. In particular, we design a proof-search algorithm (see \alg~\ref{alg:Prove} below) which takes a vocabulary $\vocab$ as a parameter and bottom-up applies rules from $\nspv$ in attempt to construct a proof of a given input sequent $\sprset \sar \phi$. 
 We may make the simplifying assumption that our proof-search algorithm only receives inputs of the form $\sprset \sar \phi$ as any nested sequent $\sprset \sar \nseq$ with $\nseq = \Sigma_{0}, (s_{1})[\Sigma_{1}], \myDots, (s_{n})[\Sigma_{n}]$ is valid \iffi $\iota(\sprset \sar \nseq)$ is valid \iffi $\sprset \sar \phi$ is valid, where $\phi = \bigvee \Sigma_{0} \lor \spbox{s_{1}}(\bigvee \Sigma_{1}) \lor \cdots \lor \spbox{s_{n}}(\bigvee \Sigma_{n})$.

To decrease the complexity of proof-search and obtain (worst-case) complexity-optimality, we introduce a new technique in the context of nested sequents which we refer to as \emph{coloring}. In essence, given the input $\sprset \sar \phi$, the first step of proof-search guesses a \emph{proper coloring} of the formula $\phi$, that is, it labels the formula's subformulae with either an \emph{active} label $\act$ or an \emph{inactive} label $\inact$ in a particular manner. Recall that, due to the $\conr$ rule, a proof in $\nspv$ has the structure of a binary tree, thus giving rise to the possibility that proof-search is exponential; 
 therefore, our proof-search algorithm uses the aforementioned labels to only generate a single path in this binary tree relative to each coloring, which yields a worst-case complexity-optimal proof-search procedure in $\conp$ (for the validity problem of $\spv$). 

\begin{definition}[Coloring]\label{def:colored-formula} We define a \emph{colored formula} to be a formula generated via the following grammar in BNF:
$$
\clr{\phi} ::= p^{\ast} \ | \ \neg p^{\ast} \ |\ (\clr{\phi} \lor \clr{\phi})^{\ast} \ | \ (\clr{\phi} \land \clr{\phi})^{\ast} \ | \ (\spdia{s} \clr{\phi})^{\ast} \ | \ (\spbox{s} \clr{\phi})^{\ast}
$$
with $\ast \in \{\circ, \bullet\}$. For any colored formula $\clr{\phi}$, we let $\phi$ be the formula in $\langv$ obtained by removing all labels $\act$ and $\inact$ from $\clr{\phi}$. A formula $\clr{\phi}$ is \emph{properly colored} \iffi $\clr{\phi} = f_{\circ}(\phi)$, where the non-deterministic \emph{coloring function} $f_{\circ}$ and $f_{\bullet}$ are defined accordingly with $\ast \in \{\circ,\bullet\}$:
\begin{itemize}

\item $\cf(p) = p^{\ast}$

\item $\cf(\neg p) = \neg p^{\ast}$

\item $\cf(\phi \lor \psi) = (\cf(\phi) \lor \cf(\psi))^{\ast}$

\item $f_{\circ}(\phi \land \psi) \in \{(f_{\circ}(\phi) \land f_{\bullet}(\psi))^{\circ}, (f_{\bullet}(\phi) \land f_{\circ}(\psi))^{\circ}\}$

\item $f_{\bullet}(\phi \land \psi) = (f_{\bullet}(\phi) \land f_{\bullet}(\psi))^{\bullet}$

\item $\cf(\spdia{s} \phi) = (\spdia{s} \cf(\phi))^{\ast}$

\item $\cf(\spbox{s} \phi) = (\spbox{s} \cf(\phi))^{\ast}$

\end{itemize}

 We define $\prc{\phi}$ to be the set of all proper colorings of $\phi$, and define a \emph{colored nested sequent} to be a nested sequent that uses colored formulae as opposed to formulae from $\langv$. 
\end{definition}

We now stipulate our \emph{saturation conditions}. When such conditions are unsatisfied during proof-search it signals that certain inference rules still need to be applied bottom-up. 
 Alternatively, once all such conditions are satisfied this signals that proof-search ought to terminate.

\begin{definition}[Saturation Conditions]\label{def:saturation-conditions} A colored nested sequent $\sprset \sep \Sigma_{0}, (s_{1})[\Sigma_{1}]_{\pi_{1}}, \myDots, (s_{n})[\Sigma_{1}]_{\pi_{n}}$ is \emph{saturated} \ifandonlyif for every $i \in \{0, \myDots, n\}$ it satisfies the following conditions:

\begin{description}[leftmargin=1.8em,style=nextline]
\item[$\scid$]  If $p^{\act} \in \Sigma_{i}$, then $\neg p^{\act} \not\in \Sigma_{i}$;
\item[$\scdis$] if $(\clr{\phi} \lor \clr{\psi})^{\act} \in \Sigma_{i}$, then $\clr{\phi}^{\act}, \clr{\psi}^{\act} \in \Sigma_{i}$;
\item[$\sccon$] if $(\clr{\phi} \land \clr{\psi})^{\act} \in \Sigma_{i}$, then either $\clr{\phi}^{\act} \in \Sigma_{i}$ or $\clr{\psi}^{\act} \in \Sigma_{i}$;
\item[$\scspdia$] if $(\spdia{s} \clr{\phi})^{\act} \in \Sigma_{i}$ and $s' \sprel s$, then for each $j \in \{1, \myDots, n\}$ such that $s_{j} = s'$, $\clr{\phi}^{\act} \in \Sigma_{j}$;
\item[$\scastdia$] if $(\spdia{\ast} \clr{\phi})^{\act} \in \Sigma_{i}$,  then $\clr{\phi}^{\act} \in \Sigma_{0}$;
\item[$\scspbox$] if $(\spbox{s} \clr{\phi})^{\act} \in \Sigma_{i}$, then for some $j \in \{1, \myDots, n\}$, $s_{j} = s$, and $\clr{\phi}^{\act} \in \Sigma_{j}$;
\item[$\scns$] for each $s \in \spset$, there exists a 
 $j \in \{1, \myDots, n\}$ \\such that $s_{j} = s$.
\end{description}
\end{definition}

\begin{algorithm}[t] 
\KwIn{A Nested Sequent: $\sprset \sar \phi$}
\KwOut{A Boolean: $\true$, $\false$}

Choose a proper coloring $\clr{\phi}$ of $\phi$;\\
\Return $\proofsearch(\sprset \vdash \clr{\phi})$;

\caption{$\prove$}\label{alg:Prove}
\end{algorithm}

\begin{algorithm}[h]
\KwIn{A Colored Nested Sequent: $\ns := \Gamma \vdash \Sigma_{0}, (s_{1})[\Sigma_{1}]_{\pi_{1}}, \myDots, (s_{n})[\Sigma_{n}]_{\pi_{n}}$}
\KwOut{A Boolean: $\true$, $\false$}

\If{for some $0\leq i\leq n$, $p^{\circ},\neg p^{\circ} \in \Sigma_{i}$}
     {\Return $\true$;}

\If{$\Sigma$ is saturated}
     {\Return $\false$;}
     
\If{for some $0\leq i\leq n$, $(\clr{\phi} \lor \clr{\psi})^{\act} \in \Sigma_{i}$, but $\clr{\phi}, \clr{\psi} \not\in \Sigma_{i}$}
{
    Let $\Sigma_{i}' := \Sigma_{i}, \clr{\phi}, \clr{\psi}$;\\
    Let $\ns' := \Gamma \! \vdash\! \Sigma_{0}, \myDots, (s_{i})[\Sigma_{i}']_{\pi_{i}}, \myDots, (s_{n})[\Sigma_{n}]_{\pi_{n}}$;\\
    \tcp{Replace $\Sigma_{i}$ by $\Sigma_{i}'$
    to obtain $\ns'$.}
    {    \Return $\prove(\ns')$;}
}

\If{for some $0\leq i\leq n$, $(\clr{\phi}^{\circ} \land \clr{\psi}^{\bullet})^{\circ} \in \Sigma_{i}$, but $\clr{\phi}^{\circ} \not\in \Sigma_{i}$}
{
    Let $\Sigma_{i}' := \Sigma_{i}, \clr{\phi}^{\circ}$;\\
    Let $\Phi' := \Gamma\! \vdash\! \Sigma_{0}, \myDots, (s_{i})[\Sigma_{i}']_{\pi_{i}}, \myDots, (s_{n})[\Sigma_{n}]_{\pi_{n}}$;\\
    \tcp{Replace $\Sigma_{i}$ by $\Sigma_{i}'$
    to obtain $\ns'$.}
    \Return $\prove(\ns')$
}

\If{for some $0\leq i\leq n$, $(\clr{\phi}^{\bullet} \land \clr{\psi}^{\circ})^{\circ} \in \Sigma_{i}$, but $\clr{\psi}^{\circ} \not\in \Sigma_{i}$}
{
    Let $\Sigma_{i}' := \Sigma_{i}, \clr{\psi}^{\circ}$;\\
    Let $\Phi' := \Gamma\! \vdash\! \Sigma_{0}, \myDots, (s_{i})[\Sigma_{i}']_{\pi_{i}}, \myDots, (s_{n})[\Sigma_{n}]_{\pi_{n}}$;\\
    \tcp{Replace $\Sigma_{i}$ by $\Sigma_{i}'$
    to obtain $\ns'$.}
    \Return $\prove(\ns')$
}

\caption{$\proofsearch$ (Part I)}\label{alg:Proof-search}
\end{algorithm}

Let us comment on the functionality of our (non-deterministic) proof-search algorithm $\prove$ (\alg~\ref{alg:Prove}), which takes $\proofsearch$ (\alg~\ref{alg:Proof-search}) as a subroutine. (NB. \alg~\ref{alg:Proof-search} is split between this page and the next due to its length.) As mentioned above, given an input $\sprset \sar \phi$, the algorithm $\prove$ guesses a proper coloring $\clr{\phi}$ of $\phi$, and then returns the value of $\proofsearch(\sprset \sar \clr{\phi})$. We note that $\proofsearch$ applies the rules from $\nspv$ in a bottom-up manner (each corresponding to a recursive call of the algorithm with the exception of $\id$). The application of each rule is as follows: $\id$ corresponds to lines 1--3, $\disr$ to lines 7--11, $\conr$ to lines 12--16 and 17--21, that respectively yields the left and right premises of $\conr$, $\spdiari{s}$ and $\spdiarii{s}$ to lines 22--25, $\spdiar{\dst}$ to lines 26--29, $\spboxr{s}$ to lines 30--33, and $\nr{s}$ to lines 34--37.

 Moreover, $\proofsearch$ contrasts with typical proof-search algorithms in that it utilizes the active and inactive labels $\act$ and $\inact$ in $\clr{\phi}$ to guide its computation and only constructs \emph{a single thread} of the proof.\footnote{A \emph{thread} in a proof is defined in the usual fashion as a path of sequents from the conclusion of the proof to an initial sequent (cf.~\cite[p.~14]{Tak13}).} In other words, if a nested sequent $\sprset \sar \phi$ is derivable in $\nspv$, then the sequent has a proof in $\nspv$ such that $\proofsearch$ generates each thread of the proof relative to each proper coloring of $\phi$; as argued in the lemma below, all such threads may be `zipped' together to reconstruct a full proof of $\sprset \sar \phi$ in $\nspv$. In this way, our proof-search algorithm may be used to construct certificates witnessing the validity (by means of a proof in $\nspv$) or invalidity (by means of a counter-model) of any input $\sprset \sar \phi$ (see \thm~\ref{thm:correctness} below for details).

\begin{lemma}\label{lem:build-proof}
 Let $\sprset \sep \phi$ be a sequent and $\prc{\phi}$ the (finite) set of proper colorings of $\phi$. If $\proofsearch(\Gamma \vdash \clr{\phi}) = \true$ for all $\clr{\phi}\in\prc{\phi}$, then there is a proof of $\Gamma \vdash \phi$ in $\nspv$. 
\end{lemma}

\begin{proof} Assume that $\proofsearch(\sprset \vdash \clr{\phi}) = \true$ for every $\clr{\phi}\in\prc{\phi}$, and the following thread of colored nested sequents is generated during its execution: 
\begin{align*}
    T(\clr{\phi}) := \sprset \sar \nseq_{0}, \myDots, \sprset \sar \nseq_{h}
\end{align*}
 such that $\nseq_{0} = \clr{\phi}$ and $\sprset \sar \nseq_{h}$ is an instance of $\id$ (by lines 1--3). Let $\mathcal{T}$ be the set of all such threads.


For a thread $\sprset \sar \nseq_{0}, \myDots, \sprset \sar \nseq_{h} \in \mathcal{T}$ and $0 \leq k \leq h$, the colored nested sequent  $\sprset \sar \nseq_{k}$ is:
\vspace{-0.65em}
\begin{itemize}
\item \emph{left conjunctive} \iffi $\sprset \sar \nseq_{k+1}$ is obtained from $\sprset \sar \nseq_{k}$ by applying $\conr$ yielding the left premise (lines 12--16), and
\item \emph{right conjunctive} \iffi $\sprset\! \sar\! \nseq_{k+1}$ is obtained from $\sprset \sar \nseq_{k}$ by applying $\conr$ yielding the right premise (lines 17--21).
\end{itemize}
\vspace{-0.5em}
\noindent A colored nested sequent is \emph{conjunctive} \iffi it is left or right conjunctive. We now explain how our threads may be transformed into a proof of $\sprset \sar \phi$ in $\nspv$.

 We assume w.l.o.g. that the initial segments of all threads up to and including the first conjunctive sequent $\sprset \sar \nseq_{k}$ are identical (i.e. we assume that the subroutine $\proofsearch$ executes in a deterministic fashion). Hence, we may form the `pseudo-derivation' shown below left, by making use of the first $k$ sequents of any given thread, where the rules $(r_{1}), \myDots, (r_{k-1})$ are determined on the basis of which lines of $\proofsearch$ were executed. 
\begin{center}
\begin{tabular}{|c|c|}
\AxiomC{$\sprset \sar \nseq_{k}$}
\RightLabel{$(r_{k-1})$}
\UnaryInfC{$\vdots$}
\RightLabel{$(r_{1})$}
\UnaryInfC{$\sprset \sar \phi$}
\DisplayProof

&

\AxiomC{$\sprset \sar \nseq_{k+1}^{l}$}
\AxiomC{$\sprset \sar \nseq_{k+1}^{r}$}
\RightLabel{$\conr$}
\BinaryInfC{$\sprset \sar \nseq_{k}$}
\RightLabel{$(r_{k-1})$}
\UnaryInfC{$\vdots$}
\RightLabel{$(r_{1})$}
\UnaryInfC{$\sprset \sar \phi$}
\DisplayProof 
\end{tabular}
\end{center}

 Let us define $T^{k}(\clr{\phi}) := \sprset \sar \nseq_{k+1}, \myDots, \sprset \sar \nseq_{h}$ to be the tail of a thread $T(\clr{\phi})$ starting from $k+1$. Since the $k^{th}$ colored nested sequent of every thread is conjunctive, we may generate two sets of threads $\mathcal{T}_{L}$ and $\mathcal{T}_{R}$ from $\mathcal{T}$:
\begin{align*}
\mathcal{T}_{L} := \{ T^{k}(\clr{\phi}) \ | \ \sprset \sar \nseq_{k} \in T(\clr{\phi}) \text{ is left conjunctive}\} \\
\mathcal{T}_{R} := \{ T^{k}(\clr{\phi})  \ | \ \sprset \sar \nseq_{k} \in T(\clr{\phi}) \text{ is right conjunctive}\}    
\end{align*}
 
We now extend the `pseudo-derivation' shown above left, with a bottom-up application of $\conr$ to obtain the `pseudo-derivation' shown above right, where $\sprset \sar \nseq_{k+1}^{l}$ and $\sprset \sar \nseq_{k+1}^{r}$ are the initial elements of  each thread $\mathcal{T}_{L}$ and $\mathcal{T}_{R}$
, respectively. By successively repeating the above described process over $\mathcal{T}_{L}$ and $\mathcal{T}_{R}$, a proof in $\nspv$ will eventually be built above $\sprset \sar \nseq_{k+1}^{l}$ and $\sprset \sar \nseq_{k+1}^{r}$, giving a proof of $\sprset \sar \phi$ in $\nspv$.
\end{proof}

\setcounter{algocf}{1}
\begin{algorithm}[h]
\setcounter{AlgoLine}{21}

\If{for some $0\leq i\leq n$, $(\spdia{s} \clr{\phi})^{\act} \in \Sigma_{i}$ and $s' \sprel s$, and for some $1\leq j \leq n$, $s_{j} = s'$, but $\clr{\phi} \not\in \Sigma_{j}$}
{
    Let $\ns' := \Gamma \! \vdash\! \Sigma_{0}, \myDots, (s_{j})[\Sigma_{j}, \clr{\phi}]_{\pi_{j}}, \myDots, (s_{n})[\Sigma_{n}]_{\pi_{n}}$;\\
    \tcp{Add $\clr{\phi}$ to the $j^{th}$ nesting
    to obtain $\ns'$.}
    {    \Return $\prove(\ns')$;}
}

\If{for some $0\leq i\leq n$, $(\spdia{\dst} \clr{\phi})^{\act} \in \Sigma_{i}$, but $\clr{\phi} \not\in \Sigma_{0}$}
{
    Let $\ns' := \Gamma \vdash\ \clr{\phi}, \Sigma_{0}, (s_{1})[\Sigma_{1}]_{\pi_{1}}, \myDots, (s_{n})[\Sigma_{n}]_{\pi_{n}}$;\\
    \tcp{Add $\clr{\phi}$ to the $0^{th}$ component
    to obtain $\ns'$.}
    {    \Return $\prove(\ns')$;}
}

\If{for some $0\leq i\leq n$, $(\spbox{s} \clr{\phi})^{\act} \in \Sigma_{i}$, but for each $1\leq j \leq n$ such that $s_{j} = s$, $\clr{\phi} \not\in \Sigma_{j}$}
{
    Let $\Sigma' := \Gamma \vdash \Sigma_{0}, \myDots, (s_{n})[\Sigma_{n}]_{\pi_{n}}, (s)[\clr{\phi}]_{\pi_{n+1}}$;\\ \tcp{Append $(s)[\clr{\phi}]_{\pi_{n+1}}$ to obtain $\ns'$ with $\pi_{n+1}$ fresh.}
    {    \Return $\prove(\ns')$;}
}

\If{for some $s \in \spset$ there does not exist a $1\leq j \leq n$ such that $s_{j} = s$}
{
    Let $\ns' := \Gamma \vdash \Sigma_{0}, \myDots, (s_{n})[\Sigma_{n}]_{\pi_{n}}, (s)[\emptyset]_{\pi_{n+1}}$;\\ \tcp{Append $(s)[\emptyset]_{\pi_{n+1}}$ to obtain $\ns'$ with $\pi_{n+1}$ fresh.}
    {    \Return $\prove(\ns')$;}
}

\caption{$\proofsearch$ (Part II)}
\end{algorithm}

\begin{example} To illustrate the procedure in \lem~\ref{lem:build-proof} above, we provide an example of the proof construction process of $\emptyset \sar \phi$ with $\phi := \spbox{s}(p \lor (\neg p \land \neg p)) \land \spdia{s} (q \lor \neg q)$. First, observe that $\phi$ has three proper colorings, given below:
\begin{description}

\item[$\clr{\phi}^{1} :=$] $((\spbox{s}(p^{\act} \lor (\neg p^{\act} \land \neg p^{\inact})^{\act})^{\act})^{\act} \land (\spdia{s} (q^{\inact} \lor \neg q^{\inact})^{\inact})^{\inact})^{\act}$

\item[$\clr{\phi}^{2} :=$] $((\spbox{s}(p^{\act} \lor (\neg p^{\inact} \land \neg p^{\act})^{\act})^{\act})^{\act} \land (\spdia{s} (q^{\inact} \lor \neg q^{\inact})^{\inact})^{\inact})^{\act}$

\item[$\clr{\phi}^{3} :=$] $((\spbox{s}(p^{\inact} \lor (\neg p^{\inact} \land \neg p^{\inact})^{\inact})^{\inact})^{\inact} \land (\spdia{s} (q^{\act} \lor \neg q^{\act})^{\act})^{\act})^{\act}$

\end{description}
Each proper coloring $\clr{\phi}^{i}$ with $i \in \{1,2,3\}$ gives rise to a corresponding thread $T(i)$ (shown at the top of \fig~\ref{fig:proof-construction-example}) when $\proofsearch(\emptyset \sar \clr{\phi}^{i})$ is run. Note that in the figure we have omitted the active and inactive labels, as well as labels of the form $\pi_{i}$, from each thread $T(i)$ to improve readability, and each application of $\conr$ is emphasized with a dashed inference line. By making use of the proof construction process described in \lem~\ref{lem:build-proof}, the three threads $T(1)$, $T(2)$, and $T(3)$ can be fused together to generate a proof of $\emptyset \sar \phi$ shown at the bottom of \fig~\ref{fig:proof-construction-example}.
\end{example}

\begin{figure*}
\begin{center}
\begin{tabular}{c c c}
$T(1) :=$ & $T(2) :=$ & $T(3) :=$ \vspace{0.5em}\\ 

\AxiomC{ }
\RightLabel{$\id$}
\UnaryInfC{$\emptyset \sar \phi, \psi_{0}, (s)[\psi_{1}, p, \neg p, \psi_{2}]$}
\RightLabel{$\conr$}
\dashedLine
\UnaryInfC{$\emptyset \sar \phi, \psi_{0}, (s)[\psi_{1}, p, \neg p \land \neg p]$}
\RightLabel{$\disr$}
\UnaryInfC{$\emptyset \sar \phi, \psi_{0}, (s)[p \lor (\neg p \land \neg p)]$}
\RightLabel{$\spboxr{s}$}
\UnaryInfC{$\emptyset \sar \phi, \spbox{s}(p \lor (\neg p \land \neg p))$}
\RightLabel{$\conr$}
\dashedLine
\UnaryInfC{$\emptyset \sar \phi$}
\DisplayProof

&

\AxiomC{ }
\RightLabel{$\id$}
\UnaryInfC{$\emptyset \sar \phi, \psi_{0}, (s)[\psi_{1}, p, \neg p, \psi_{2}]$}
\RightLabel{$\conr$}
\dashedLine
\UnaryInfC{$\emptyset \sar \phi, \psi_{0}, (s)[\psi_{1}, p, \neg p \land \neg p]$}
\RightLabel{$\disr$}
\UnaryInfC{$\emptyset \sar \phi, \psi_{0}, (s)[p \lor (\neg p \land \neg p)]$}
\RightLabel{$\spboxr{s}$}
\UnaryInfC{$\emptyset \sar \phi, \spbox{s}(p \lor (\neg p \land \neg p))$}
\RightLabel{$\conr$}
\dashedLine
\UnaryInfC{$\emptyset \sar \phi$}
\DisplayProof

&

\AxiomC{ }
\RightLabel{$\id$}
\UnaryInfC{$\emptyset \sar \phi, \spdia{s} (q \lor \neg q), (s)[\psi_{3}, q, \neg q]$}
\RightLabel{$\disr$}
\UnaryInfC{$\emptyset \sar \phi, \spdia{s} (q \lor \neg q), (s)[q \lor \neg q]$}
\RightLabel{$\spdiar{s}$}
\UnaryInfC{$\emptyset \sar \phi, \spdia{s} (q \lor \neg q), (s)[\emptyset]$}
\RightLabel{$\nr{s}$}
\UnaryInfC{$\emptyset \sar \phi, \spdia{s} (q \lor \neg q)$}
\RightLabel{$\conr$}
\dashedLine
\UnaryInfC{$\emptyset \sar \phi$}
\DisplayProof
\end{tabular}
\end{center}

 \vspace{0.5em}

\begin{center}

\AxiomC{ }
\RightLabel{$\id$}
\UnaryInfC{$\emptyset \sar \phi, \psi_{0}, (s)[\psi_{1}, p, \neg p, \psi_{2}]$}
\AxiomC{ }
\RightLabel{$\id$}
\UnaryInfC{$\emptyset \sar \phi, \psi_{0}, (s)[\psi_{1}, p, \neg p, \psi_{2}]$}
\RightLabel{$\conr$}
\BinaryInfC{$\emptyset \sar \phi, \psi_{0}, (s)[\psi_{1}, p, \neg p \land \neg p]$}
\RightLabel{$\disr$}
\UnaryInfC{$\emptyset \sar \phi, \psi_{0}, (s)[p \lor (\neg p \land \neg p)]$}
\RightLabel{$\spboxr{s}$}
\UnaryInfC{$\emptyset \sar \phi, \spbox{s}(p \lor (\neg p \land \neg p))$}
\AxiomC{ }
\RightLabel{$\id$}
\UnaryInfC{$\emptyset \sar \phi, \spdia{s} (q \lor \neg q), (s)[\psi_{3}, q, \neg q]$}
\RightLabel{$\disr$}
\UnaryInfC{$\emptyset \sar \phi, \spdia{s} (q \lor \neg q), (s)[q \lor \neg q]$}
\RightLabel{$\spdiar{s}$}
\UnaryInfC{$\emptyset \sar \phi, \spdia{s} (q \lor \neg q), (s)[\emptyset]$}
\RightLabel{$\nr{s}$}
\UnaryInfC{$\emptyset \sar \phi, \spdia{s} (q \lor \neg q)$}
\RightLabel{$\conr$}
\BinaryInfC{$\emptyset \sar \phi$}
\DisplayProof
\end{center}

\caption{An example of how the threads $T(1)$, $T(2)$, and $T(3)$ may be `zipped' together to construct the derivation shown above bottom. Note that $\phi := \spbox{s}(p \lor (\neg p \land \neg p)) \land \spdia{s} (q \lor \neg q)$, $\psi_{0} := \spbox{s}(p \lor (\neg p \land \neg p))$, $\psi_{1} := p \lor (\neg p \land \neg p)$, $\psi_{2} := \neg p \land \neg p$, and $\psi_{3} := q \lor \neg q$.}\label{fig:proof-construction-example}
\end{figure*}

\begin{theorem}[Correctness]\label{thm:correctness}
 Let $\sprset \sep \phi$ be a sequent.
\begin{enumerate}
    \item If $\proofsearch(\sprset \sep \clr{\phi}) = \true$ for all proper colorings of $\phi$, then a proof in $\nspv$ may be constructed witnessing that $\sprset \sep \phi$ is valid. 
    
    \item If $\proofsearch(\sprset \sep \clr{\phi}) = \false$ for some proper coloring of $\phi$, then a counter-model may be constructed witnessing that $\sprset \sep \phi$ is invalid.
\end{enumerate}
\end{theorem}

\begin{proof} The first claim follows by \lem~\ref{lem:build-proof} and the soundness of each nested calculus (see \thm~\ref{thm:soundness}); therefore, we focus on the second claim. Suppose that $\proofsearch(\sprset~\sep~\clr{\phi}) = \false$ for some chosen proper coloring of $\phi$. Then, $\proofsearch$ generates a saturated nested sequent $\sprset \sep \nseq$ with $\nseq$ of the form $\Sigma_{0}, (s_{1})[\Sigma_{1}]_{\pi_{1}}, \myDots, (s_{n})[\Sigma_{n}]_{\pi_{n}}$. We will use $\sprset \sep \nseq$ to construct a counter-model for $\sprset \sep \phi$, thus proving the second claim. Let us define $\model := \langle \Pi, \sigma, \delta \rangle$ as:
\SetLabelAlign{CenterWithParen}{\hfil#1\hfil}
\begin{itemize}
    \item $\Pi := \{\pi_{0}, \pi_{1}, \myDots, \pi_{n}\}$;
    \item $\sigma(\dst) := \Pi$ and for each $s \in \spset \setminus \{\dst\}$, we define $\sigma(s)$:
\begin{enumerate}[label=(\roman*),leftmargin=1em,labelwidth=1em,align=CenterWithParen]
    \item $\pi_{i}\in\sigma(s)$ for each $i \in \{1, \myDots, n\}$ such that $s_{i} = s$, and 
    \item $\pi_{j}\in\sigma(s)$ for each $\pi_{j} \in \sigma(s')$ 
    such that $s' \sprel s$;
\end{enumerate} 
    \item $\delta(p) := \{\pi_{i} \ | \ p^{\act} \not\in \Sigma_{i}\}$.
\end{itemize}
We know that the recursive definition of $\sigma$ will eventually terminate since $\spset$ is finite.

Let us now prove that $\model$ is indeed a standpoint model; afterward, we will show that $\model, \pi_{0} \not\models \iota(\sprset \sep \nseq)$. First, observe that the inclusion of $\pi_{0}$ in $\Pi$ ensures that the domain of precisifications is non-empty. Furthermore, $\delta$ is a function from $\pset$ to $2^{\Pi}$, for each $s \in \spset$, we have that $\sigma(s) \neq \emptyset$ by the $\scns$ saturation condition, and $\sigma(\dst) = \Pi$ by definition. We now show that (a) for each $s' \preceq s \in \Gamma$ and $i \in \{0, \myDots, n\}$, $\model, \pi_{i} \models s' \preceq s$, and (b) for each $i \in \{0, \myDots, n\}$ and $\clr{\psi} \in \Sigma_{i}$, $\model, \pi_{i} \not\models \psi$, from which $\model, \pi_{0} \not\models \iota(\sprset \sep \phi)$ follows as $\clr{\phi} \in \Sigma_{0}$.

(a) Let $s' \preceq s \in \Gamma$ 
 and assume that $\pi \in \sigma(s')$.  Then, by the second clause in the definition of $\sigma(s)$ it follows that $\pi \in \sigma(s)$, thus showing that $\model, \pi \models s \preceq s'$ for all $\pi \in \Pi$.

(b) By induction on the complexity of $\psi$.

\textit{Base case.} First, let us suppose that $\psi$ is a propositional atom $p$ where $p^{\act} \in \Sigma_{i}$. Then, by the definition of $\delta$, we have that $\pi_{i} \not\in \delta(p)$, i.e. $\model, \pi_{i} \not\models p$. Similarly, if $\psi$ is a negated atom $\neg p$ where $\neg p^{\circ} \in \Sigma_{i}$, then by the $\scid$ saturation condition, we know that $p^{\circ} \not\in \Sigma_{i}$, implying that $\pi_{i} \in \delta(p)$ by the definition of $\delta$, thus showing that $\model, \pi_{i} \not\models \neg p$. 

\textit{Inductive step.} We suppose that $\clr{\psi}^{\act} \in \Sigma_{i}$ and we show that $\model, \pi_{i} \not\models \psi$ making a case distinction based on the main connective of $\psi$. 
\vspace{-0.5em}
\begin{description} 
     
\item[$\psi = \chi \lor \xi$:] By the $\scdis$ saturation condition, $\clr{\chi}^{\act}, \clr{\xi}^{\act} \in \Sigma_{i}$. By IH, $\model, \pi_{i} \not\models \chi$ and $\model, \pi_{i} \not\models \xi$, showing $\model, \pi_{i} \not\models \psi$.

\item[$\psi = \chi \land \xi$:] By the $\sccon$ saturation condition, $\clr{\chi}^{\act} \in \Sigma_{i}$ or $\clr{\xi}^{\act} \in \Sigma_{i}$. By IH, $\model, \pi_{i} \not\models \chi$ or $\model, \pi_{i} \not\models \xi$, so  $\model, \pi_{i} \not\models \psi$.

\item[{\rm $\psi = \spdia{s} \chi \text{ with } s \in \spset \setminus \{\dst\}$}:]  
Assume that $\pi_{j} \in \sigma(s)$. By the definition of $\sigma(s)$, $\pi_{j}$ was added via the first (i) or second (ii) condition. If (i), then there is a nesting $(s)\{\Sigma_{j}\}_{\pi_{j}}$ in $\nseq$, and by the $\scspdia$ saturation condition, $\clr{\chi}^{\act} \in \Sigma_{j}$. By IH then, $\model, \pi_{j} \not\models \chi$. If (ii), then there is a chain of statements $s_{0} \preceq s_{1}, \myDots, s_{n} \preceq s \in \Gamma$ such that a nesting of the form $(s_{0})\{\Sigma_{j}\}_{\pi_{j}}$ exists in $\nseq$. Again, by the $\scspdia$ saturation condition, it follows that $\clr{\chi}^{\act} \in \Sigma_{j}$, which implies that $\model, \pi_{j} \not\models \chi$ by IH. Hence, $\model, \pi_{i} \not\models \spdia{s} \chi$.

\item[$\psi = \spdia{\ast} \chi$:] Since $\dst$ is maximal relative to the $\sprel$ relation, i.e. $s \sprel \dst$ for all $s \in \spset$, we have $\clr{\chi}^{\act} \in \Sigma_{j}$ for each $j \in \{0, \myDots, n\}$ by the $\scastdia$ and $\scspdia$ saturation conditions, hence for all $\pi_{j} \in \Pi$, $\model, \pi_{j} \not\models \chi$, i.e. $\model, \pi_{i} \not\models \spdia{\dst} \chi$.

\item[$\psi = \spbox{s} \chi$:] By the $\scspbox$ saturation condition, we know that there exists a $j \in \{1, \myDots, n\}$ such that $s_{j} = s$ and $\clr{\chi}^{\act} \in \Sigma_{j}$. By IH, $\model, \pi_{j} \not\models \chi$, thus proving the case.\qedhere
\end{description}
\end{proof}

\begin{example} Let us provide an example of the counter-model construction procedure given in \thm~\ref{thm:correctness}. We assume that the (invalid) sequent $s \preceq s' \sar \phi$ with $\phi := \spbox{s'}p \lor \spdia{s}\neg p$ is input into $\prove$.  Since $\phi$ has one proper coloring (with all subformulae active), only the following single thread is generated, yielding the saturated sequent shown at the top of the proof below. We omit the active labels for readability.
\begin{center}
\resizebox{\columnwidth}{!}{
\AxiomC{$s \preceq s' \sar \phi, \spbox{s'}p, \spdia{s}\neg p, (s')[p]_{\pi_{1}}, (s)[\neg p]_{\pi_{2}}, (*)[\emptyset]_{\pi_{3}}$}
\RightLabel{$\nr{*}$}
\UnaryInfC{$s \preceq s' \sar \phi, \spbox{s'}p, \spdia{s}\neg p, (s')[p]_{\pi_{1}}, (s)[\neg p]_{\pi_{2}}$}
\RightLabel{$\spdiari{s}$}
\UnaryInfC{$s \preceq s' \sar \phi, \spbox{s'}p, \spdia{s}\neg p, (s')[p]_{\pi_{1}}, (s)[\emptyset]_{\pi_{2}}$}
\RightLabel{$\nr{s}$}
\UnaryInfC{$s \preceq s' \sar \phi, \spbox{s'}p, \spdia{s}\neg p, (s')[p]_{\pi_{1}}$}
\RightLabel{$\spboxr{s'}$}
\UnaryInfC{$s \preceq s' \sar \phi, \spbox{s'}p, \spdia{s}\neg p$}
\RightLabel{$\disr$}
\UnaryInfC{$s \preceq s' \sar \phi$}
\DisplayProof
}
\end{center}

\vspace{0.5em}

Then, we may extract the following (counter-)model $\model = \langle \Pi, \sigma, \delta \rangle$ from the top, saturated sequent in the proof above.
\vspace{-1em}
\begin{itemize}
    \item $\Pi := \{\pi_{0}, \pi_{1},\pi_{2},\pi_{3}\}$;
    \item $\sigma(\dst) := \Pi$, $\sigma(s) := \{\pi_{2}\}$, and $\sigma(s') :=  \{\pi_{1}\}$;
    \item $\delta(p) := \{\pi_{0}, \pi_{2}, \pi_{3}\}$. 
\end{itemize}

\noindent It is readily verifiable that $\model, \pi_{0} \not\models \spbox{s'}p \lor \spdia{s} \neg p$.
\end{example}

We now show that $\proofsearch$ (and hence $\prove$) terminates after at most polynomially many rule applications in the \emph{size} of the input sequent. For an input $\ns := \sprset \sar \phi$, its size is defined to be $\size{\ns} := \card{\spset} + \len{\phi}$. That is, the size of $\ns$ is the sum of the cardinality of the set $\spset$ of standpoints and the 
 size of $\phi$. The size of a sequent incorporates a measure on the set $\spset$ from the associated vocabulary $\vocab$ as opposed to a measure on the set $\sprset$ of sharpening statements because $\sprset$ only plays a role in bottom-up applications of $\spdiari{s}$ and $\spdiarii{s}$, which are bounded in part by the cardinality of $\spset$ and in part by the number of $\spbox{s}$ modalities occurring in $\phi$, as explained in the proof of \thm~\ref{thm:termination} below. 

\begin{theorem}[Termination]\label{thm:termination}
Let $\ns := \sprset \sep \phi$ be a sequent. Then, the number of recursive calls in $\proofsearch(\sprset \sep \clr{\phi})$, and thus $\prove(\sprset \sep \phi)$, is bounded by a polynomial

\vspace{-0.8em}
$$
p(\size{\ns}) = \bigo(|\ns|^{2}).
$$
\end{theorem}

\vspace{-0.8em}
\begin{proof} Let $\ns := \sprset \sep \phi$ be a nested sequent, and $N_{\oplus}$ be the number of occurrences of the connectives $\{\lor,\land\} \cup \{\spdia{s} \ | \ s \in \spset\}$ in $\phi$. By the saturation conditions (\dfn~\ref{def:saturation-conditions}), we know that for each $s \in \spset$, the  $\spboxr{s}$ rule will be applied bottom-up at most one time for each occurrence of $\spbox{s}$ in $\phi$, which are bounded by $|\phi|$. Also, $\nr{s}$ will be applied at most once for each $s \in \spset$. Since only $\spboxr{s}$ and $\nr{s}$ introduce nestings, the number of components (i.e. the nestings plus the root) throughout the course of proof-search is bounded by:

\vspace{-0.8em}
$$
K := 1 + |\spset| + |\phi|
$$

For each occurrence of $\lor$, $\land$, and $\spdia{s}$ in $\phi$ (with $s \in \spset$), we know by the saturation conditions that $\disr$, $\conr$, $\spdiari{s}$, $\spdiarii{s}$, and $\spdiar{*}$ can be applied a maximum number of $K$ times during proof search. Then, since $N_{\lor} + N_{\land} + \sum_{s \in \spset} N_{\spdia{s}}\leq|\phi|$, the number of recursive calls (i.e. bottom-up applications of rules) during proof-search is bounded by $N := |\phi| \cdot K$.
Finally, $\card{\spset}, \len{\phi} \leq \size{\ns}$ holds trivially, implying:

\vspace{-0.8em}
$$
N \leq \size{\ns} \cdot (1 + \size{\ns} + \size{\ns} ) 
$$

Therefore, it follows that a polynomial $p(\size{\ns}) = \bigo(\size{\ns}^{2})$ bounds the number of recursive calls of $\prove(\ns)$. 
\end{proof}

\begin{corollary}\label{cor:decid-FMP-complexity} Let $\vocab$ be a vocabulary. Then,
\begin{enumerate}
    \item $\spv$ is decidable;
    \item $\spv$ has the finite model property;
    \item $\prove$ is worst-case complexity-optimal, deciding the validity problem for $\spv$ in $\conp$;
    \item The validity problem for $\spv$ is $\conp$-complete.
\end{enumerate}
\end{corollary}

\begin{proof} Statements 1 and 2 follow from the fact that $\prove$ is a correct (\thm~\ref{thm:correctness}) and terminating (\thm~\ref{thm:termination}) decision procedure for $\spv$ that, in particular, returns a finite counter-model when the input is invalid. 

To show statement 3, observe that $\prove$ is a non-deterministic algorithm that takes a sequent $\ns := \sprset \sar \phi$ as input, guesses a proper coloring of $\phi$, and constructs a thread. Each such thread is polynomial in the size of its input, since the number of rule applications (i.e. the length of the thread) is bounded by a polynomial $p(\size{\ns}) = \bigo(\size{\ns}^{2})$, by \thm~\ref{thm:termination}. Moreover, since any sequent generated during proof-search can have at most $K \leq 1 + |\spset| + |\phi|$ many components (as stated in the proof of \thm~\ref{thm:termination}), each of which can only be inhabited by at most $\card{\sufo{\phi}} = \len{\phi}$ many formulae, it follows that the size of each nested sequent in the thread is bounded by $\bigo(|\ns|^{2})$ since $\card{\spset}, \len{\phi} \leq \size{\ns}$. Taking the functionality of $\prove$ into account, one can see that if $\prove(\ns) = \false$, then the corresponding thread is generated in polynomial time and its size is bounded above by a polynomial $q(\size{\ns}) = \bigo(\size{\ns}^{4})$. Additionally, note that $\prove$ is worst-case complexity-optimal as the validity problem for classical propositional logic is $\conp$-complete, and can be solved by $\prove$ as $\id$, $\disr$, and $\conr$ form a sound and complete proof system for propositional logic (cf.~\cite[\app~B]{Lyo21thesis}).
Last, statement 4 is an immediate consequence of statement 3.
\end{proof}

\section{Conclusion and Future Work}\label{sec:conclusion}

In this paper, we introduced and employed nested sequent systems to automate reasoning with propositional standpoint logics. To obtain worst-case complexity-optimal proof-search, we presented a novel proof-search technique, referred to as \emph{coloring}, whereby the subformulae of an input formula are non-deterministically colored with (in)active labels, yielding partial proofs (i.e. \emph{threads}) of the input. By means of our technique, we designed a non-deterministic proof-search algorithm deciding the validity of standpoint implications in $\conp$, showing how certain threads could be transformed into a counter-model for an invalid input, and how all threads could be transformed into a proof for a valid input. The attainment of these ``certificates'' from proof-search serve as explanations for the (in)validity of standpoint formulae, thus motivating our proof-theoretic approach.

For future work, we aim to extend our nested systems and proof-search algorithm to cover (i) first-order standpoint logics that (ii) incorporate complex standpoints, which have interesting applications in knowledge integration scenarios. Regarding point (i), placing standpoint logic on a first-order base increases the applicability of the framework along with its expressivity to better match that of contemporary knowledge representation languages. Our focus in this area is to provide results that can then be extrapolated to widely used decidable fragments of FOL.
Regarding point (ii), we note that the set-theoretic interpretation of standpoints permits the definition of complex standpoints built atop atomic ones; e.g. union $s_{1} \cup s_{2}$ (integrating knowledge from multiple perspectives), intersection $s_{1} \cap s_{2}$ (expressing the knowledge jointly shared between multiple perspectives), and difference $s_{1} \setminus s_{2}$ (yielding the sharpening of $s_{1}$ by ignoring all precisfications of $s_{2}$). Beyond providing nested systems for more expressive formulations of standpoint logic, we also aim to write and evaluate theorem provers based on our nested calculi. 


\noindent
\acknowledgments{Lucía Gómez Álvarez was supported by the Bundesministerium fur Bildung und Forschung (BMBF, Federal Ministry of Education and Research) in the Center for Scalable Data Analytics and Artificial Intelligence (ScaDS.AI). Tim S. Lyon has received funding from the European Research Council (Grant Agreement no. 771779, DeciGUT). }


\bibliographystyle{kr}
\bibliography{bibliography,LuciaRefs}




\end{document}